\documentclass[12pt]{amsart}
\usepackage{diagbox}
\usepackage{mathtools}
\mathtoolsset{showonlyrefs}

\renewcommand{\emph}[1]{\textit{#1}}
\usepackage{enumerate,amsmath,amsthm,latexsym,amssymb}
\usepackage{color}\usepackage{graphicx}

\definecolor{brown}{cmyk}{0, 0.72, 1, 0.45}
\definecolor{grey}{gray}{0.5}

\def\dK{\vec{K}}
\newcommand{\old}[1]{}

\newcounter{rot}%\addtocounter{rot}{1}, \therot

\def\hr{\widehat{\rho}}
\newcommand{\ignore}[1]{}

\def\cA{{\mathcal A}}

\def\cS{{\mathcal S}}

\newcommand{\set}[1]{\left\{#1\right\}}

\def\ii_(#1,#2){i_{#1}^{#2}}

\def\a{\alpha}
\def\b{\beta}
\def\d{\delta}

\def\e{\varepsilon}
\def\f{\phi}

\def\g{\gamma}
\def\G{\Gamma}
\def\k{\kappa}

\def\z{\zeta}

\def\th{\theta}

\def\l{\lambda}
\def\m{\mu}
\def\n{\nu}

\def\r{\rho}

\def\s{\sigma}
\def\t{\tau}
\def\om{\omega}

\def\cE{\mathcal{E}}
\def\cF{\mathcal{F}}

\newcommand{\brac}[1]{\left( #1 \right)}

\newcommand{\expect}{\operatorname{\bf E}}
\def\E{\expect}
\def\Var{{\bf Var}}
\renewcommand{\Pr}{\operatorname{\bf Pr}}
\newcommand\bfrac[2]{\left(\frac{#1}{#2}\right)}

\newtheorem{theorem}{Theorem}[section]

\newtheorem{lemma}[theorem]{Lemma}
\newtheorem{corollary}[theorem]{Corollary}

\newtheorem{remthm}[theorem]{Remark}

\newenvironment{remark}{\begin{remthm}\rm }{\end{remthm}}%
\newcounter{thmtemp}

\newcommand{\nospace}[1]{}

\def\path{\operatorname{PATH}}

\newcommand{\beq}[2]{\begin{equation}\label{#1}#2\end{equation}}

\begin{document}
\title{Online purchasing under uncertainty}%Probabilistic analysis of purchasing problems}
\author{Alan Frieze}
\thanks{Research supported in part by NSF Grants DMS1362785, CCF1522984 and a grant(333329) from the Simons Foundation}
\author{Wesley Pegden}
\thanks{Research supported in part by NSF grant
    DMS1363136}
\address{Department of Mathematical
  Sciences,\\Carnegie Mellon University,\\Pittsburgh PA 15213.}

\begin{abstract}
Suppose there is a collection $x_1,x_2,\dots,x_N$ of independent uniform $[0,1]$ random variables, and a hypergraph $\cF$ of \emph{target structures} on the vertex set $\{1,\dots,N\}$.  We would like to purchase a target structure at small cost, but we do not know all the costs $x_i$ ahead of time.  Instead, we inspect the random variables $x_i$ one at a time, and after each inspection, choose to either keep the vertex $i$ at cost $x_i$, or reject vertex $i$ forever.

In the present paper, we consider the case where $\{1,\dots,N\}$ is the edge-set of a complete graph (or digraph), and the target structures are the spanning trees of a graph, spanning arborescences of a digraph, the paths between a fixed pair of vertices, perfect matchings, Hamilton cycles or the cliques of some fixed size.  
\end{abstract}

\maketitle

\section{Introduction}
Suppose we inspect independent and uniform $[0,1]$ random variables $x_1,x_2,\dots,x_N$ one at a time.  After each $i$th inspection, we decide to pay the cost $x_i$ to purchase the index $i$, or pass. If we \emph{must} purchase some $i$, what is the minimum expected cost of an optimal strategy?  This problem is closely related to the well-studied and well-generalized secretary problem (see for example \cite{secprob}, \cite{BIK} and \cite{KK}), and is attributed by Ferguson \cite{Fer} to Cayley \cite{Cay}. It was solved by Moser \cite{Mos}. (In Section \ref{s.lists}, we will see that an optimum strategy pays asymptotically $\frac 2 N$ in expectation.)

The focus of the present paper concerns the more general problem, where we must purchase not just one $i$, but a set $I$ belonging to some \emph{target hypergraph} $\cF$ on the vertex set $\{1,\dots,N\}$.  We query the costs $x_i$ in some order, and since we imagine that querying $x_i$ makes its availibility unstable in some way, we must again decide on each step whether to purchase the queried $x_i$, or pass on $x_i$ forever.  We are interested in the minimum expected cost paid by an optimal algorithm which always succeeds at purchasing every element of at least one hyper-edge of $\cF$, which we call the \emph{purchase price} of $\cF$.  

Presumably, the order in which the $x_i$'s are examined should have a large effect on the expected price paid for structures.  We will in fact concern ourselves with three distinct models, in which we have progressively less control over the order of inspection.

{\bf Purchaser Ordered Online Model -- POM:} In this model, at each step,  we are allowed to inspect any uninspected $x_i$ to see if we wish to purchase it. It is therefore an on-line model where we choose the order of the items.

{\bf Randomly Ordered Online Model -- ROM:} In this model the $x_i$ are presented to us in random order $x_{i_1},\dots,x_{i_N}$.  At each step $t$ we learn the index $i_t$ and the associated cost $x_{i_t}$, and must decide immediately whether or not to purchase $e_i$.

{\bf Adversarially Ordered Online Model -- AOM} In this model, the order of inspections is determined by an {\em adaptive} adversary. (The adversary does not know the costs of uninspected $x_i$'s.)

Note that when evaluating the purchase price of $\cF$, the expected cost is computed just over a probability space of weights for POM and AOM, but in a product space of weights and edge-orderings, for ROM.

Surprisingly, for the problems we consider, we achieve have upper bounds in ROM which are close to our POM lower bound. In particular, the effect of optimum control over the order of inspection (versus random order) is small for these problems. We will also show that some of the upper bounds in ROM can also be achieved even in AOM. On the other hand, there will be cases where there is a substantial gap between what can be achieved in these models.

For most of the paper $x_i$ will be the cost of the edge $e_i$ of the complete graph (or digraph) on $n$ vertices, in which case $N=\binom{n}{2}$ (or $n(n-1)$). Each hyper-edge of $\cF$ will correspond to a particular desired graphical structure. As an example, suppose that the hyper-edges of $\cF$ correspond to the edge-sets of paths between vertices $1$ and $2$. A minimum cost hyper-edge of $\cF$ corresponds to a shortest path between vertices $1$ and $2$ in $K_n$. We know that w.h.p.\footnote{A sequence of events $\cE_n,\,n=1,2,\ldots$ is said to occur with high probability (w.h.p.) if $\lim_{n\to\infty}\Pr(\cE_n)=1$.} this is asymptotically $\frac{\log n}{n}$\cite{Janson}. We will prove (see Theorem \ref{thm:spaths}) that even in the POM framework, an optimal strategy to purchase such a path pays much more in expectation.

We will begin our discussion though with something closely related to the famous ``secretary problem''. In section \ref{s.lists} we consider the particularly simple case of determining the purchase price $\r_{N,k}$ where $\cF$ is the complete $k$-uniform hypergraph on $N$ vertices; in other words, the task is simply to purchase any $k$ items. (Note that the $S_n$ symmetry of the hypergraph means that POM and ROM are equivalent for this hypergraph.)  This case will arise as a tool in many of our more complicated analyses. We call this the {\em $k$-purchase problem}.
\begin{theorem}\label{rhoNk}
For $k\geq 1$ we have
\beq{ckdef}{
\r_{N,k}\approx\frac{c_k}{N}\text{ where $c_1=2$ and $c_k=c_{k-1}+1+\sqrt{1+2c_{k-1}}$ for $k\geq 2$.}
}
\end{theorem}
Here we use the notation $a_n\approx b_n$ to indicate that $\lim_{n\to \infty}a_n/b_n=1$.

The following corollary will be needed in the proof of Theorem \ref{TSP}. We only need it for $k=1$, although we claim it for arbitrary $k$.
\begin{corollary}\label{cor1}
If we replace uniform $[0,1]$ by a distribution that has a density $f(x)=Dx+o(x)$ as $x\to 0$ then we simply replace $\frac{c_k}{N}$ by $\frac{Dc_k}{N}$.
\end{corollary}
With these tools available, we examine the cost of finding a shortest path between two specified vertices.
\begin{theorem}\label{thm:spaths}
The POM and ROM purchase prices of a path between vertices $1$ and $2$ in $K_n$ are both $n^{-2/3\pm o(1)}$. The AOM price is $\Omega(n^{-1/2})$.
\end{theorem}
This stands in stark contrast to the length $\approx \tfrac{\log n}{n}$ we could achieve if we were allowed to examine all the edges ahead of time.

A similarly stark increase in cost occurs when we aim for the seemingly simple target of a triangle.  Although for arbitrary $\om\to\infty$, a uniform $[0,1]$ weighting of $K_n$ has a triangle of cost $O\bfrac{\om}{n},$ w.h.p., we prove:
\begin{theorem}\label{triangles} The POM and ROM purchase prices of a triangle in $K_n$ are both $n^{-4/7\pm o(1)}$. The AOM price is $\Omega(n^{-1/2})$.
\end{theorem}

We can generalize the ROM upper bound in Theorem \ref{triangles} to the problem of purchasing a copy of a clique $K_r$. We will prove that if $\kappa^{\mathrm{ROM}}_{r,n}$ is the expected cost of purchasing a copy of $K_r,r\geq 3$ in ROM, then
\begin{theorem}\label{Krthm}
$$\kappa_{r,n}^{\mathrm{ROM}}\leq O\bfrac{1}{n^{d_r+o(1)}}$$
$d_3=\frac47,d_4=\frac29$ and $d_r=\frac{1}{11\cdot 2^{r-5}-1},i\geq 5$.
\end{theorem}
We note that the distribution of the minimum cost of a clique when are allowed to examine all the edges is of the order $n^{-2/(k-1)}$, see \cite{FPS}.

Despite these examples, purchasing a structure in our model is not always so prohibitively costly, compared with the minimum weight structure available.  In particular, recall that the minimum cost spanning tree for a uniform $[0,1]$ weighting of $K_n$ has asymptotic length $\zeta(3)$ w.h.p. \cite{FriezeTree}
\begin{theorem}\label{thMST}
Let $\b^{\text{POM}}_n,\b^{\text{ROM}}_n$ be the purchase price of a spanning tree in the POM and ROM models, respectively.   Then, for $n$ large, we have
\[
\z(3)< 1.38\leq \b^{\text{POM}}_n\leq \b^{\text{ROM}}_n<2\z(3).\\
\]
Furthermore, the AOM price is $O(1)$.
\end{theorem}
We will give a short proof that $\b^{\text{POM}}_n\geq 1.25$ in Section \ref{STLB}. The bound $\b^{\text{POM}}_n\geq 1.38$ is due to prior work of Aldous, Angel and Berestycki \cite{AAB} which is incomplete.   (Although their paper considers the ROM problem, the proof of their lower bound is valid for $\b^{\text{POM}}_n$ also.)  We have included a sketch of their proof in Section \ref{STLB1}.

For our next result, let us consider the setting where $\cF$ is the hypergraph of perfect matchings on the edge-set of the complete bipartite graph $K_{n,n}$.  
\begin{theorem}\label{BPM}  Letting $\m^{\text{POM}}_n,\m^{\text{ROM}}_n,\m^{\text{AOM}}_n$ be the POM, ROM and AOM purchase prices of perfect matchings on $K_{n,n}$, we have
\[
2\lesssim \m^{\text{POM}}_n\leq \m^{\text{ROM}}_n\lesssim \m^{\text{AOM}}_n\lesssim 4c_3,\]
\end{theorem}
where $c_3$ is as in Theorem \ref{rhoNk}. Here $a\lesssim b$ denotes $a\leq (1+o(1))b$.

After this we turn to the non-bipartite case. 
\begin{theorem}\label{PM}
If $\nu^{\text{POM}}_n,\nu^{\text{ROM}}_n,\nu^{\text{AOM}}_n$ are the POM, ROM and AOM purchase prices of a perfect matching in $K_n$, we have 
\[
2\lesssim \nu^{\text{POM}}_n\leq \nu^{\text{ROM}}_n\lesssim \nu^{\text{AOM}}_n\lesssim 4c_3.\]
\end{theorem}

We can also purchase a Traveling Salesperson tour at constant expected cost:
\begin{theorem}\label{TSP}
If $h^{\text{POM}}_n,h^{\text{ROM}}_n,h^{\text{AOM}}_n$ are the POM, ROM and AOM purchase prices of a Hamilton cycle in $K_n$, then
$$c_2\lesssim h^{\text{POM}}_n\leq h^{\text{ROM}}_n\lesssim h^{\text{AOM}}_n\lesssim 200.$$
\end{theorem}
In Theorems \ref{thMST}, \ref{BPM}, \ref{PM}, \ref{TSP}, we prove the purchase prices are $\Theta(1)$, but we do not determine an asymptotic constant.  In our final example, we can determine such a constant.

In particular, we consider the case where the underlying set of $\cF$ is the set of $n(n-1)$ directed edges in the complete digraph $\dK_n$, and the target sets are the \emph{arborescences}; i.e., directed rooted spanning trees, such that every edge is oriented towards the root.
\begin{theorem}\label{tharb} If $\a_n^{\text{POM}},\a_n^{\text{ROM}}$ are the POM and ROM purchase prices of an arborescence in $\dK_n$,  then 
\[
\lim_{n\to \infty}\a_n^{\text{POM}}=\lim_{n\to \infty}\a_n^{\text{ROM}}=2.
\]
Furthermore, in the AOM model we can w.h.p. construct an arborescence of cost $O(1)$.
\end{theorem}
Note that the w.h.p result does not necessarily imply an $O(1)$ purchase price for the arborescence in the AOM model.  An interesting open question is whether $\lim_{n\to \infty}\a_n^{\text{AOM}}=2$ as well.

As a final result, we prove w.h.p results for the Traveling Salesperson Problem (TSP) on the complete digraph $\vec{K}_n$. Again, it is not clear that this implies results for the purchase price (expected value).
\begin{theorem}\label{DTSP}
In the directed case of the TSP we have that w.h.p. 
$$4\lesssim H^{\text{POM}}_n\leq H^{\text{ROM}}_n\lesssim H^{\text{AOM}}_n\lesssim 4c_2,$$
where $H^{model}$ is the cost of the tour we can find w.h.p in the model.
\end{theorem}
\section{Preliminaries}
For reference we use the following: Let $B(n,p)$ denote the binomial random variable with parameters $n,p$. Then for $0\leq\e\leq 1$ and $\a>0$, we have the following Chernoff bounds:
\begin{align}
\Pr(B(n,p)\leq (1-\e)np)&\leq e^{-\e^2np/2}.\label{chern1}\\
\Pr(B(n,p)\leq (1+\e)np)&\leq e^{-\e^2np/3}.\label{chern2}\\
%\Pr(B(n,p)\geq \a np)&\leq\bfrac{e}{\a}^{\a np}.\label{chern3}
\end{align}
We will also make use of a simple application of the Azuma-Hoeffding martingale tale inequality, often referred to as McDiarmid's inequality \cite{McD}. Let $Z=Z(Y_1,Y_2,\ldots,Y_N)$ where $Y_1,Y_2,\ldots,Y_N$ are independent random variables. Suppose that if ${\bf Y}=(Y_1,Y_2,\ldots,Y_N)$
and ${\bf Y'}=(Y_1',Y_2',\ldots,Y_N')$ differ in only one coordinate that $|Z({\bf Y})-Z({\bf Y}')|\leq c$. Then for any $u>0$ we have
\beq{MCD}{
\Pr(|Z-\E Z|\geq u)\leq 2\exp\set{-\frac{u^2}{2Nc^2}}.
}
\section{Purchasing any $k$ items}\label{s.lists}
The symmetry of the $k$-uniform hypergraph means that the order in which we inspect the variables $x_i$ is irrelevant.  Thus, we suppose they are simply given in some fixed order $x_N,x_{N-1},\ldots,x_1$. In other words, in this context, there is essentially no difference between any of the three models under consideration. At step $i$ (beginning with $i=N$), we examine $x_i$, and either \emph{accept} $i$ and pay $x_i$, or \emph{reject} $i$ and continue to step $i-1$.  The process ends after $k$ acceptances.  Note that since we are required to purchase $k$ items, we will have to pay all of $x_{\ell},x_{\ell-1},\dots,x_1$ to accept $\ell,\dots,1$ if only $k-\ell$ items have been accepted before index $\ell$ is reached.

Our goal is to minimize the expected total cost; we denote this expected value by $\rho_{N,k}$.
\subsection{The case $k=1$}\label{k1} 
We write $\rho_{N}:=\rho_{N,1}$ for the expected cost in this case.  When we inspect $x_N$, an optimum strategy to minimize our total expected cost is to accept $x_N$ if and only if $x_N<\rho_{N-1}$, since $\rho_{N-1}$ is the expected cost when we reject $x_N$.  This simple dynamic will allow us to prove the case $k=1$ of Theorem \ref{rhoNk} by analyzing the process recursively.
\begin{proof}
Let $Z$ denote the cost of the element selected by the optimal strategy so that $\r_N=\E(Z)$. Then we have the following recurrence:
\begin{align}
\r_N&=\E(Z\mid x_N\leq \r_{N-1})\Pr(x_N\leq \r_{N-1})+\E(Z\mid x_N> \r_{N-1})\Pr(x_N> \r_{N-1})\label{E}\\
&=\frac{\r_{N-1}^2}{2}+\r_{N-1}(1-\r_{N-1})\\
&=\r_{N-1}\brac{1-\frac{\r_{N-1}}{2}}.\label{rhorec}
\end{align}
We can now easily get an upper bound on $\r_N$ by induction. Clearly, $\r_1=1/2$ and if $\r_{N-1}\leq 2/N$ for some $N>1$ then we see from \eqref{rhorec} that,
\beq{upperrho}{
\r_N\leq \frac{2}{N}\brac{1-\frac{1}{N}}\leq \frac{2}{N+1}.
}
For a corresponding lower bound let $\e_N=A/N$ where $A=10$ suffices. We have checked numerically that $\r_N\geq 2(1-\e_{N-1})/(N-1)$ for $N\leq 3(A+1)=33$. Now assume that $N\geq 3(A+1)$ and $\r_{N-1}\geq2(1-\e_{N-1})/(N-1)$. Then, because $x(1-x/2)$ increases monotonically in $[0,1]$,
\begin{align}
\r_N&\geq \frac{2(1-\e_{N-1})}{N-1}-\frac{(1-\e_{N-1})^2}{(N-1)^2}\\
&=\frac{2(1-\e_{N})}{N}+\frac{2(\e_N-\e_{N-1})}{N-1}+\frac{2(1-\e_N)}{N(N-1)} -\frac{(1-\e_{N-1})^2}{(N-1)^2}\\
&=\frac{2(1-\e_{N})}{N}-\frac{2A}{N(N-1)^2}+ \frac{2(N-A)}{N^2(N-1)}-\frac{(N-1-A)^2}{(N-1)^4}\\
&\geq \frac{2(1-\e_{N})}{N}-\frac{2A}{N(N-1)^2}+ \frac{2(N-A)}{N^2(N-1)}-\frac{(N-A)^2}{N^2(N-1)^2}\\
&=\frac{2(1-\e_{N})}{N}+\frac{N^2-2(A+1)N-(A^2-2A)}{N^2(N-1)^2}\\
&\geq \frac{2(1-\e_{N})}{N},
\end{align}
since $N\geq 3(A+1)$.
\end{proof}
\subsection{The case $k\geq2$}\label{more}
We now complete the proof of Theorem \ref{rhoNk}.
\begin{proof}
We replace \eqref{E} by 
\begin{align}
\r_{k,N}&=\E\left(\min\set{x_N+\r_{k-1,N-1},\r_{k,N-1}}\label{Ek}\right)\\
&=\E\left(x_N+\r_{k-1,N-1}\mid x_N<\r_{k,N-1}-\r_{k-1,N-1}\right)\Pr(x_N<\r_{k,N-1}-\r_{k-1,N-1})\\
&\hspace{2cm}+\rho_{k,N-1}\Pr(x_N\geq \r_{k,N-1}-\r_{k-1,N-1})\\
&=\brac{\frac{\r_{k,N-1}-\r_{k-1,N-1}}{2}+\r_{k-1,N-1}}(\r_{k,N-1}-\r_{k-1,N-1})+\r_{k,N-1}(1-(\r_{k,N-1}-\r_{k-1,N-1}))\\
&=\r_{k,N-1}-\frac{(\r_{k,N-1}-\r_{k-1,N-1})^2}{2}.\label{rhoreck}
\end{align}
Suppose inductively that $\r_{k,N-1}\leq \frac{c_k}{N},\,\r_{k-1,N-1}\leq \frac{c_{k-1}}{N}$. Then, \eqref{rhoreck} and $\r_{k,N-1}\geq \r_{k-1,N-1}$ implies that
\begin{align}
\r_{k,N}&=\r_{k,N-1}\brac{1-\frac{\r_{k,N-1}}{2}+\r_{k-1,N-1}}-\frac{\r_{k-1,N-1}^2}{2}\label{must}\\
&\leq \frac{c_k}{N}-\frac{(c_k-c_{k-1})^2}{2N^2}\\
&\leq\frac{c_k}{N+1}+\frac{2c_k-(c_k-c_{k-1})^2}{2N^2}\\
&=\frac{c_k}{N+1}.
\end{align}
Here the form of \eqref{must} implies that we should put $\r_{k,N-1},\r_{k-1,N-1}$ at their upper bounds to maximize the RHS.

For a lower bound, let $\e_N=1/N^{1/2}$ and assume inductively that 
$$\r_{k,N-1}\geq \brac{1-\e_{N-1}}c_k/(N-1),\r_{k-1,N-1}\geq \brac{1-\e_{N-1}}c_{k-1}/(N-1).$$ 
We will first prove that
\beq{cklow}{
\r_{k,k}=\frac{k}{2}\geq \brac{1-\frac{1}{k^{1/2}}}\frac{c_k}{k}
}
for $k\geq 2$, which will allow us a base for our induction.

If we let $1+2c_k=d_k^2$ then \eqref{ckdef} implies that 
\beq{dkdef}{
d_k^2=1+2c_k=1+2(c_{k-1}+1+d_{k-1})=d_{k-1}^2+2d_{k-1}+2=(d_{k-1}+1)^2+1.
}
Since $d_1>1$, this implies that 
$$d_k>k.$$
On the other hand, \eqref{dkdef} and $(1+x)^{1/2}\leq x^{1/2}+\frac{1}{2x}$ implies that
$$d_k\leq d_{k-1}+1+\frac{1}{2(d_{k-1}+1)^2}\leq d_{k-1}+1+\frac{1}{2k^2}.$$
It follows from this that
$$d_k\leq k+\sum_{r=1}^\infty \frac{1}{2r^2}<k+1\text{ which implies that }c_k\leq \frac{k^2}{2}\brac{1+\frac{2}{k}}.$$
We can see therefore that \eqref{cklow} holds if $\brac{1+\frac{2 }{k}}\brac{1-\frac{1}{k^{1/2}}}\leq 1$. This holds for $k\geq 2$ and the proof of \eqref{cklow} is complete.

We now use induction on $N$ with $N=k$ as the base case. Referring to \eqref{rhoreck}, the function
 $f(x,y)= x-(x-y)^2/2$ is monotone increasing in $x$ for $x \le 1+y$ and in $y$ for  $y \le x$. Thus we can substitute our lower bounds for $\r_{k, N-1}$ and $\r_{k-1,N-1}$ and use $2c_k=(c_k-c_{k-1})^2$ to get
\begin{align*}
\r_{k,N}  &\ge \brac{1-\e_{N-1}}\frac{c_k}{N-1}-
\frac{\brac{1-\e_{N-1}}^2}{2(N-1)^2}(c_k-c_{k-1})^2\\
&=
\brac{1-\e_{N-1}}\frac{c_k}{N-1}-\frac{(1-\e_{N-1})^2c_k}{(N-1)^2}
\\
&=\frac{(1-\e_N)c_k}{N}+\frac{(\e_N-\e_{N-1})c_k}{N}+ \frac{(1-\e_{N-1})c_k}{N(N-1)}-\frac{(1-\e_{N-1})^2c_k}{(N-1)^2}\\
&= \frac{(1-\e_N)c_k}{N} + c_k g(N),
\end{align*}
where
\begin{align}
g(x)&=\frac{\sqrt{1/x}-\sqrt{1/(x-1)}}{x}+\frac{1-\sqrt{1/(x-1)}}{x(x-1)}
-\frac{(1-\sqrt{1/(x-1)})^2}{(x-1)^2}\\
&=\frac{1-\sqrt{1+\frac{1}{x-1}}}{x^{3/2}}+ \bfrac{1-\frac{1}{\sqrt{x-1}}}{x-1} \brac{\frac{1}{x}-\frac{1-\frac{1}{\sqrt{x-1}}}{x-1}}\\
&\geq -\frac{1}{2x^{3/2}(x-1)}+\frac{(\sqrt{x-1}-1)(x-\sqrt{x-1})}{x(x-1)^{5/2}}\\
&\geq \frac{x(2\sqrt{x-1}-5)+2+\sqrt{x-1}}{2x^{5/2}(x-1)}\\
&\geq 0\text{ for }x\geq 7.
\end{align}
By direct evaluation $\r_{k,N}\geq (1-\e_N)c_k/N$ for $2\leq k<N\leq 6$.
\end{proof}
We will need the following estimates for $c_k$ as $k$ grows:
\begin{lemma}\label{cklemma}
If $k\geq 1$ then
$$\frac{k^2}{2}\leq c_k\leq 2k^2.$$
\end{lemma}
\begin{proof}
$c_1=2$ and then we inductively have for $k>1$,
\begin{align*}
c_k&\leq 1+2(k-1)^2+\sqrt{1+4(k-1)^2}\\
&=1+2(k-1)^2+2(k-1)\brac{1+\frac{1}{4(k-1)^2}}^{1/2}\\
&\leq 1+2(k-1)^2+2(k-1)\brac{1+\frac{1}{8(k-1)^2}}\\
&\leq 2k^2-2k+\frac{5}{4}\\
&\leq 2k^2.
\end{align*}
This confirms the upper bound.

Similarly,
\begin{align*}
c_k&\geq 1+\frac{(k-1)^2}{2}+\brac{1+(k-1)^2}^{1/2}\\
&\geq 1+\frac{k^2}{2}-k+\frac12+k-1\\
&\geq \frac{k^2}{2}.
\end{align*}
\end{proof}
\subsection{Purchasing at least one and two on average}
In this section, we consider a peculiar variant of the purchasing problem.  In this variant, we consider uniform costs $x_1,\dots,x_N$ as before, and still must decide on the spot to accept or reject an index.  Now, however, the total number of accepted indices is not required to be fixed.  Instead, as a random variable (depending on the costs $x_1,\dots,x_N$), the number $\nu_N$ of accepted indices is required to satisfy:
\begin{enumerate}
\item $\nu_N\geq 1$
\item $\E(\nu_N)=2$.
\end{enumerate}

We call this the {\em average-two-purchase problem} and we consider this problem merely for technical reasons; it will arise in our analysis of the purchase price of spanning trees in Section \ref{STLB}. Let $\phi_{N}$ denote the minimum expected total cost of items purchased under these constraints.  (The minimum is taken over all strategies which ensure that $\nu_N$ satisfies the two conditions above.)
\begin{theorem}\label{average}
$$\phi_N\gtrsim \frac{2.5}{N}.$$
\begin{proof}
The first item accepted must cost $\gtrsim 2/N$ in expectation, else we can improve the upper bound for $k=1$ in Theorem \ref{rhoNk}. 

Let $1\leq T\leq N$ be the random variable equal to the step on which the first item is accepted.  We are required to purchase, on average, one item from among $x_{T+1},x_{T+2},\dots,x_{N}$; let $C$ denote the optimum expected cost to do this.  ($C$ depends on the distribution of $T$.)

Evidently $C$ is at least the optimum expected cost to purchase an average of one item from $x_1,\dots,x_N$, and an optimum strategy for this task is to accept any item seen which has cost $\leq \tfrac 1 N$.  Indeed, this strategy accepts an average of exactly one item, and each rejection of an item of cost $\leq \tfrac 1 N$ requires that the strategy is modified to sometimes accept some items of cost greater than $\frac 1 N$ to compensate; so that the expected total cost $\frac{1}{2N}$ can only be increased.  Thus $C\geq \tfrac{1}{2N}$, and consequently,
$$\phi_N\gtrsim \frac{2}{N}+\frac{1}{2N}.$$
\end{proof}
\end{theorem}
\subsection{Proof of Corollary \ref{cor1}}
Using $\hr_{k,N}$ to denote the expected minimum cost in this context, we write 
\begin{multline*}
\hr_{k,N}=\brac{\frac{\hr_{k,N-1}-\hr_{k-1,N-1}}{2}+\hr_{k-1,N-1}}\brac{D(\hr_{k,N-1}-\hr_{k-1,N-1})+O\bfrac{k}{N}}\\
+\hr_{k,N-1}\brac{1-D(\hr_{k,N-1}+\hr_{k-1,N-1})+O\bfrac{k}{N}}.
\end{multline*}
We then carry out the analysis of Section \ref{more} with $\r_{N,k}$ replaced by $D\hr_{N,k}$, making adjustment for the error terms.
\section{Shortest Paths}\label{SP}
In this section, we prove Theorem \ref{thm:spaths}. We again contrast this with the fact that Janson \cite{Janson} proved that offline, the shortest distance is $\approx\frac{\log n}{n}$ w.h.p.
\subsection{Upper Bound}
We compute our upper bound in the ROM model. Let $\a=2/3$ and $k=1/\log\log n$. We partition the edge set of $K_n$ into $X,Y,Z$ where $X$ is the first $N/3$ edges and $Y$ is the next $N/3$ edges and $Z$ is the final $N/3$ edges in the given order. Then partition $X$ into $X_1,X_2,\ldots,X_k$ where each set is of size $N/(3k)$ and where $X_{i+1}$ follows $X_i$ in the random order. Similarly partition $Y$ into $Y_1,Y_2,\ldots,Y_k$. Next let $p=n^{-1+\a/k}$. Now let $G(X_i)$ denote the graph induced by edges in $X_i$ that are of cost at most $p$, $i=1,2,\ldots,k$. Define $G(Y_i),i=1,2,\ldots,k$ similarly, as well as $G(Z)$. 

Let $S_0=\set{1}$ and $S_1'$ be the set of neighbors of vertex 1 in $G(X_1)$. We can determine $S_1'$ in the ROM model. Let $\e=n^{-\a/3k}$ and $S_1$ be the first $(1-\e)nq$ members of $S_1'$, where $q=p/3k$, assuming that there are at least this many neighbors. Observe that $|S_1'|$ is distributed as $Bin(n-1,q)$ and so from the Chernoff bound \eqref{chern1} we see that for $\ell=1$ we have 
\beq{Sl}{
\Pr(|S_\ell'|\leq ((1-\e)nq)^\ell)\leq \exp\set{-\frac{nq}{3n^{2\a/3k}}}=o(n^{-3}).
}
We now inductively define sets $S_\ell$ of size $((1-\e)nq)^\ell$ for $\ell=2,\ldots,k$. Given $S_\ell$, we let $S_{\ell+1}'$ be the set of vertices not in $S_{\leq\ell}=\bigcup_{i=1}^\ell S_i$ which are neighbors of $S_\ell$ in $G(X_{\ell+1})$. Now $|S_{\ell+1}'|$ is distributed as the binomial $Bin(n-m_1,(1-(1-q)^{m_2})$, where $m_1=|S_{\leq\ell}|$ and $m_2=|S_\ell|$. Then since $(1-q)^{m_2}\leq 1-m_2q+(m_2q)^2$ we see that $|S_{\ell+1}'|$ has expectation at least $(n-m_1)m_2q(1-m_2q)/3k$. Putting $1-\d=\frac{(1-\e)nqm_2}{(n-m_1)m_2q(1-m_2q)}\leq 1-\e/2$ (using $\ell<k$) we again see from \eqref{chern1} that \eqref{Sl} holds with $\ell$ replaced by $\ell+1$. In which case we, can let $S_{\ell+1}$ be the first $((1-\e)nq)^{\ell+1}$ members of $S'_{\ell+1}$. This completes the induction and we see that the construction of $S_1,S_2,\ldots,S_k$ succeeds with probability $1-o(n^{-3})$.

Let $E_\ell$ be a set of $|S_{\ell+1}|$ edges joining $S_\ell$ to $S_{\ell+1}$ in $\G_\a$. Let $T$ be the tree with vertex set $\bigcup_{i=1}^kS_\ell$ and edge set $\bigcup_{\ell=1}^{k-1}E_\ell$. 

$T$ can be constructed in our model. We simply check the edges between $S_\ell$ and $[n]\setminus S_{\leq \ell}$ as we see them and accept edges if they (i) have cost at most $p$, and (ii) connect $S_\ell$ to a new vertex, (iii) we have accepted fewer than $((1-\e)nq)^{\ell+1}$ edges of this sort. The cost of the edges of $T$ will be 
$$O((nq+(nq)^2+\cdots+(nq)^k)p) =O(n^{\a+\a/k-1}).$$ 
If $n\in T$ then we are done. Otherwise we build a corresponding tree rooted $T'$ rooted at $n$.

Because $|S_k|,|S_k'|\geq \bfrac{1-\e}{3k}^kn^\a=n^{\a-o(1)}$ w.h.p., we see that w.h.p. there will be at least $|S_k||S_k'|/4$ edges joining $S_k,S_k'$ in $G(Z)$. Then, with failure probability at most\\
$\brac{1-\frac{(\log n)^2}{|S_k||S_k'|}}^{|S_k||S_k'|/4}=o(n^{-3})$ we can find an edge in $Z$ between $S_k$ and $S_k'$ of cost at most $(\log n)^2/|S_k|\,|S_k'|=O(n^{-2\a+o(1)})$. Putting things together, we get a total cost of
$$O(n^{\a+\a/k-1})+O(n^{-2\a})=O(n^{-2/3+O(1/k)}),$$
with probability $1-o(n^{-3})$. 

We have explained the algorithm in terms of events that happen w.h.p. To get a bound on expectation, we have to also deal with the unlikely cases. So, we always check to see that the next edge to be considered is essential for connecting vertex 1 to vertex $n$. if it is, we take this edge and an arbitrary path $P$ joining these two vertices. The cost of $P$ is at most $n$ and we need to construct it only with probability $o(n^{-3})$. This adds $o(n^{-2})$ to the expectation and completes our proof of the upper bound.
\subsection{Lower Bound -- POM} 
We now compute a lower bound in the POM model. Consider an algorithm for finding a short path from vertex 1 to vertex $n$.  We let $C_1(t)$ and $C_2(t)$ be the components of the purchased graph at step $t$ which contain the vertices $1$ and $n$, respectively. (Thus, for each $t$, $C_1(t)$ and $C_2(t)$ are random variables depending on $x_1,\dots,x_{\binom{n}{2}}$.

Observe that any algorithm which finds a path from 1 to $n$ finishes its task by choosing an edge $e^*$ from $C_1(t)$ to $C_2(t)$ for some step $t$.  We will thus analyze the cost of building components $C_1(t)$ and $C_2(t)$ in terms of their size, and also the cost of the final edge $e^*$.  We begin by considering the cost of $e^*$.

{\bf Cost of $e^*$:}  To bound the cost of $e^*$ from below, we will even give the algorithm extra information, revealing the costs of all edges between $C_1(t)$ and $C_2(t)$ as soon as their endpoints belong to the respective components; in particular, we will assert that even the minimum cost edge between these two sets is not too small.  As such, from the standpoint of choosing an inexpensive $e^*$, we view the algorithm's task as simply producing by time $t_0$ components $C_1(t_0),C_2(t_0)$ between which there exists an edge $e^*$ of small cost.

We thus let $e_1,e_2,\dots,e_\ell=e^*$ be a sequence of edges consisting of all edges between $C_1(t_0)$ and $C_2(t_0)$, ordered such that if $e_i$ appears before $e_j$ (in the sense that $e_i$ but not $e_j$ joins the sets $C_1(t)$ and $C_2(t)$ for some $t$) then $i<j$.  This is a sequence of independent (and unconditioned) uniform $[0,1]$ random variables.  We will simply argue that no term $e_j$ in this sequence can be much less than $\tfrac 1 j$. 

Indeed, for any sequence $\xi_1,\xi_2,\ldots,$ of independent uniform $[0,1]$ random variables, we have for any $n$ that
\beq{infseq}{
\Pr\brac{\exists k:k\xi_k\leq \frac{1}{(\log \max\set{k,n})^2}}\leq \sum_{k=1}^\infty \frac{1}{k(\log \max\set{k,n})^2}=O\bfrac{1}{\log n}.
}
Thus w.h.p. we have that the minimum $m_\ell=\min_{j\leq \ell}\{e_j\}$ satisfies
\beq{1overC}{
m_\ell\geq \frac{1}{\ell\log^2 n}\geq \frac{1}{|C_1(t_0)||C_2(t_0)|(\log n)^2}.
}

{\bf Cost of $C,C'$:} For this we argue that if $\cE_m$ denotes the event that there is a sub-tree of $K_n$ with at least $m\gg\log n$ vertices and total cost $\leq m/3n$ then

\beq{subtree}{\Pr(\cE_m)\leq \sum_{k\geq m}\binom{n}{k}k^{k-2}\frac{(m/3n)^{k-1}}{k!}\leq \frac{3n}{m}\sum_{k\geq m}\frac{1}{k^2}\bfrac{me}{3k}^k=O\bfrac{ne^m}{3^mm^3}=o(1).
}
{\bf Explanation:} We choose our tree of $K_n$ with $k$ vertices in $\binom{n}{k}k^{k-2}$ ways. Then we use the fact that if $\eta_1,\eta_2,\ldots,\eta_M$ are independent $[0,1]$ random variables and $\z\leq 1$ then $\Pr(\eta_1+\eta_2+\cdots+\eta_M\leq \z)= \z^M/M!$.

It follows from \eqref{1overC} and \eqref{subtree} that w.h.p. the algorithm must pay at least
$$\min_{ab\geq n^{1/2}}\set{\frac{1}{ab(\log n)^2}+\frac{a+b}{3n}}=\Omega\bfrac{1}{n^{2/3}(\log n)^2}.$$
Here $a=|C|,b=|C'|$ and we minimize over $ab\geq n^{1/2}$, say, because if $ab\leq n^{1/2}$ then we already pay $n^{-1/2-o(1)}$ in expectation from \eqref{1overC}.
\subsection{Lower Bound -- AOM}
We now compute a lower bound in the AOM model. Our strategy is to place the edges $X$ incident with vertices $1,n$ last in the sequence. Suppose now that after we have seen all edges except those in $X$ we have purchased components $C_1,C_2,\ldots,C_m$ in the subgraph induced by the vertices $[2,n-1]$. The expected cost of purchasing an edge between 1 and $C_j$ is at least $(2-o(1))/|C_j|$, by Theorem \ref{rhoNk}. We can assume therefore that at least one component is of size at least $n^{1/2}$. But then we have that w.h.p. 
\beq{tc1}{
\text{every tree of size $n^{1/2}$ has cost at least $1/10n^{1/2}$.} 
}
Indeed, let $k=n^{1/2}$ and $C_n=1/10n^{1/2}$. Then, if $p(T)$ denotes the price of tree $T$,
\beq{cT}{
\Pr(\exists T:\;p(T)\leq C_n)\leq \binom{n}{k}k^{k-2}\frac{C_n^{k-1}}{(k-1)!}\leq \frac{1}{C_nk}\bfrac{ne^2C_n}{k}^k=o(1).
}
{\bf Explanation:} The factor $\frac{C_n^{k-1}}{(k-1)!}$ is the probability that the sum of $k-1$ independent uniform $[0,1]$ random variables is at most $C_n$.

It follows that it costs at least $(1-o(1))/10n^{1/2}$ in expectation to purchase a path from 1 to $n$ in the AOM model.

This completes the proof of Theorem \ref{thm:spaths}.
\section{Triangles and Paths of Length Two}
In this section, we prove Theorem \ref{triangles}. We first observe that the expected number of triangles of total length $\l$ is bounded by $\binom{n}{3}\frac{\l^3}{6}\approx\frac{(\l n)^3}{6}$. And then another calculation shows via the Chebyshev inequality that if $\l n\to \infty$ then there will be a triangle of cost $\l$ w.h.p. As we will see, it will not be possible to find such a triangle w.h.p., instead we will show that if $\tau^{\mathrm{POM}}_n,\tau^{\mathrm{ROM}}_n$ are the expected minimum cost of a triangle constructible in the POM and ROM settings, then
\[
\frac{1}{n^{4/7}\log n}\leq \tau^{\mathrm{POM}}_n\leq \tau^{\mathrm{ROM}}_n\leq \frac{10}{n^{4/7}}.
\]
Our proof of this inequality will follow from an analysis of the problem of purchasing a large collection of paths of length 2; we will see that such a feat is sufficient, and in some sense also necessary, to purchase a low-cost triangle.
\subsection{Paths of length Two}
We will prove the following:
\begin{theorem}\label{PLT}
Let $\k^{\mathrm{POM}}_{\ell,n},\k^{\mathrm{ROM}}_{\ell,n}$ denote the expected cost of purchasing $\ell=o(n)$ paths of length two in the POM and ROM settings, respectively. Then
$$\bfrac{\ell}{16n(\log n)^4}^{4/3}\leq \k^{\mathrm{POM}}_{\ell,n}\leq \k^{\mathrm{ROM}}_{\ell,n}\leq 6\bfrac{\ell}{n}^{4/3}.$$
\end{theorem}
\subsubsection{Obtaining $\ell=o(n)$ paths of length two: Upper Bound }\label{ell1}
Our goal in this section is to show how to find $\ell$ distinct paths of length two.
We will subsequently use the case $k=1$ of Theorem \ref{rhoNk} to close one of these to a triangle at expected cost order $1/\ell$. 

The first $N/3$ edges will be considered to be colored red, the next $N/3$ will be considered to be colored blue and the final $N/3$ edges will be considered to be colored green. We will use a parameter $k$ and its value will be revealed shortly. We choose $k=o(n)$ disjoint red edges by examining the red edges and accepting the first $k$ of cost at most $10k/n^2$. Then we choose $\ell$ blue edges incident with the selected red edges by examining the relevant blue edges and accepting the first $\ell$ edges of cost at most $2\ell/kn$. The Chernoff bounds imply that we will succeed with superpolynomial probability if $k,\ell=n^{\Omega(1)}$. Success will mean the creation of $\ell$ paths of length two. 

The expected cost of creating $\ell$ paths of length two with this strategy is at most 
\beq{tricost}{
\frac{5k^2}{n^2}+\frac{2\ell^2}{kn}.
}
We choose $k$ to minimize \eqref{tricost}. This gives $k=(\ell^2n/5)^{1/3}\gg\ell$ and a total cost of at most $6(\ell/n)^{4/3}$, as claimed in Theorem \ref{PLT}.

If we ever get the stage where it is impossible to obtain either a triangle or $\ell$ paths of length two without the inclusion of the next edge, regardless of cost, then we abort the above procedure and choose any triangle still available. This will cost at most 3 and the effect on the expectation is negligible.

We will assume from now on in our approach to proving an upper bound, that we will always check as to whether or not the next edge must be chosen in order to be able to construct the required object. In which case, we take it and build an object regardless of cost. It will be apparent that this happens with such low probability that it makes a negligible contribution to the expectation.
\begin{remark}\label{remx}
If we replace the underlying graph $K_n$ by $G_{n,p}$, $p$ constant then it is straightforward to see that we can carry out the same construction at an extra cost of a factor $1/p$ in expectation and w.h.p.
\end{remark}
\subsubsection{Obtaining $(\log n)^5\leq \ell=o(n)$ paths of length two: Lower Bound}
We consider the cost of producing $\ell$ paths of length two.  One natural way of producing many paths of length two is to create a collection of edge disjoint stars (i.e., trees with only one vertex of degree greater than 1).  Most of our analysis in this section is aimed at constraining (from below) the cost of purchasing such a collection of stars.  We begin by arguing that this will suffice to give a lower bound in general; i.e., that any strategy to produce many paths of length 2 is in some sense not too far from a strategy to produce many disjoint stars.

To make this reduction, let us consider a modified game: in particular, let us suppose that whenever we choose to purchase an edge, we must also assign it an orientation.  Our goal is to purchase (and orient) a set of edges to ensure that
\[
\cS=\sum_{v} \binom{\mathrm{outdeg}(v)}{2}
\]
is large.  Observe that any strategy in the modified game which achieves that $\cS=\ell$ for some $\ell$ can be translated into a strategy in the standard game which purchases $\ell$ paths of length 2 at the same expected cost, while purchasing only edge disjoint stars.  In particular, our analysis later will thus give a lower bound on the cost in the modified game of achieving that $\cS=\ell$.

On the other hand, we argue that for any algorithm $\cA$ for the standard game to purchase $\ell$ edges at small expected cost, there is a strategy in the modified game to ensure that $\cS\geq \ell/16.$  Indeed, to see this, we consider two cases.\\
{\bf Case 1: $\cA$ produces at least $\ell/2$ paths from large stars.}\\
Here a large star is a star with more than $L=(\log n)^2$ edges.
Let there be $d_i$ stars with $i$ edges. We have that
$$\sum_{i=L+1}^nd_i\binom{i}{2}\geq \frac{\ell}{2}.$$
In conjunction with Lemma \ref{largestar} below, this implies that the cost is at least 
$$\sum_{i=L+1}^n\frac{d_ii^2}{3n}\geq \frac{\ell}{4n}\gg \bfrac{\ell}{n}^{4/3},$$
and the expected cost is already above our lower bound.\\
{\bf Case 2: $\cA$ produces at least $\ell/2$ paths from small stars.}\\
We translate $\cA$ into a strategy for the modified game by orienting each purchased edge randomly.  A path of length two is said to be good if both edges are oriented away from the center vertex, so that $\cS$ is simply the number of good paths. The expected number of good paths is at least $\ell/8$. Also, switching the orientation of an edge can only change the number of good paths by at most $2L$. It follows from McDiarmid's inequality \eqref{MCD} that 
$$\Pr(\cS\leq \ell/16)\leq 2\exp\set{-\frac{\ell^2}{128\ell L^2}}.$$
\bigskip

This reduction allows us to focus on strategies in the standard game which restrict themselves to purchasing edge disjoint stars.  As a first step, we rule out the relevance of large stars, by showing that all large stars are quite expensive (even offline).
\begin{lemma}\label{largestar}
W.h.p. every star with $k>L=(\log n)^2$ edges has cost at least $\xi_k=k^2/3n$.
\end{lemma}
\begin{proof}
The probability that there is a star of smaller cost than claimed is at most
$$n\sum_{k=L+1}^{n-1}\binom{n-1}{k}\frac{\xi_k^k}{k!}\leq n\sum_{k=L+1}^{n-1}\bfrac{ne\xi_k}{k^2}^k=o(1).$$
We have used the fact that if $Z_1,\ldots,Z_r$ are independent uniform $[0,1]$ random variables then $\Pr(Z_1+\cdots+Z_r\leq \th)\leq \th^r/r!$.
\end{proof}

Suppose next that there are $k_i,1\leq i\leq L\leq \ell^{1/2}$ edges such that when the algorithm purchases them they are added to a star with $i-1$ edges already for $i=1,2,\ldots,L$. Call these {\em type $i$ edges}. If $i=1$ then then these edges are the first of their stars. When $i\geq 2$ and such an edge is added, $i-1$ new paths of length two are created. Thus 
\beq{stars}{
\frac{\ell}{16}\leq \sum_{i=2}^L (i-1)k_i\leq \ell+L\leq 2\ell,
} 
where the upper bound follows from the fact that we can stop purchasing edges once we have at least $\ell$. 

The expected cost of constructing these stars is at least 
\beq{stcost}{
a_0\sum_{i=1}^L\frac{k_i^2}{k_{i-1}n},
}
where $k_0=n$ and $a_0$ is an absolute constant. 

{\bf Explanation of \eqref{stcost}:} We can assume that the $k_i$ edges of type $i$ are chosen before any edges of type $i+1$. This provides the largest choice for each edge and thus gives a lower bound. That said, the expected cost of adding $k_i$ edges to $k_{i-1}$ vertices, by choosing from $\Omega(k_{i-1}n)$ edges is as claimed in \eqref{stcost}. 

We are therefore left with considering the following optimization problem, where $x_i=k_i/n$ and $\l=\ell/16n$: 
\beq{opt}{
\text{Minimize }\sum_{i=1}^L\frac{x_i^2}{x_{i-1}}\text{ subject to }\sum_{i=2}^L (i-1)x_i\geq \l \text{ and } 1=x_0\geq x_1\geq x_2\geq \cdots\geq x_L.
}
Observe now that 
\beq{new1}{
\sum_{i=2}^L (i-1)x_i\leq L^2x_2\text{ which implies that }x_2\geq \frac{\l}{L^2}.
}
Going back to \eqref{opt}, we have a lower bound of
\beq{new2}{
\text{Minimum: } x_1^2+\frac{x_2^2}{x_1}\text{ subject to }x_2\geq \frac{\l}{L^2}.
}
So optimizing with respect to $x_1$ for a given value of $x_2$ we get a lower bound of
$$x_2^{4/3}\brac{2^{-2/3}+2^{1/3}}\geq \bfrac{\l}{L^2}^{4/3}=\bfrac{\ell}{16L^2n}^{4/3},$$
as required in Theorem \ref{PLT}.

\subsubsection{Creating a Triangle: Upper Bound} 
We use the green edges and the case $k=1$ of Theorem \ref{rhoNk} to find a triangle by selecting a low cost edge joining the paths of length two. Notice that by construction, no two paths of length two have the same endpoints. Thus the expected cost of creating a triangle is at most
\beq{closing}{
6\bfrac{\ell}{n}^{4/3}+\frac{6}{\ell}.
}
We have a cost of $6/\ell$ because only 1/3 of the paths of length two can be completed by a green edge. Optimizing our choice of $\ell=(3/4)^{3/7}n^{4/7}$ in \eqref{closing} gives us the required upper bound in Theorem \ref{triangles}. 
\begin{remark}\label{remy}
As in Remark \ref{remx} we can replace $K_n$ by $G_{n,p}$, $p$ constant, at a cost of a factor $1/p$ in expectation and w.h.p.
\end{remark}
\subsubsection{Creating a Triangle: Lower Bound in POM}
For the lower bound we suppose that our algorithm creates $\ell$ paths of length two before adding an edge that closes a triangle. This gives us a lower bound of
$$\bfrac{\ell}{16L^2n}^{4/3}+\frac{1}{\ell(\log n)^2}\geq \frac{1}{n^{4/7}(\log n)^2}.$$
Note that the term $\frac{1}{\ell(\log n)^2}$ arises as in \eqref{infseq}.
\subsection{Creating a Triangle: Lower Bound in AOM}
The adversary's strategy is to first present all of the edges incident with vertex 1. Suppose that we accept the edges $\set{1,v},v\in A$ where $|A|=k$. The expected cost of these edges is $\approx c_k/n$. The adversary will now present all edges within $A$. We have two choices now. We can accept one of the edges within $A$ and create a triangle at expected cost
$$\Omega\brac{\frac{c_k}{n}+\frac{2c_1}{k^2}}=\Omega\brac{\frac{k^2}{n}+\frac{1}{k^2}}= \Omega\bfrac{1}{n^{1/2}}. $$
Failing this we will have to build our triangle without using vertex 1. The adversary now presents the edges incident with vertex 2 that have not been presented, and so on. We can assume inductively that this costs $\Omega\bfrac{1}{(n-1)^{1/2}}$ in expectation and we are done.

This completes the proof of Theorem \ref{triangles}.
\section{Complete Graphs Upper Bound in ROM}
Recall that $A_{r,n}$ is the expected cost of purchasing a copy of $K_r,r\geq 3$ and we will prove that $A_{r,n}=O(n^{-d_r+o(1)})$. We have already proved that we can take $d_3=4/7$ and we start an induction from here. 

%We will prove that Theorem \ref{Krthm} holds from the recurrence
%\beq{Arn}{
%A_{r+1,n}\leq \min_\ell \set{\frac{(4+o(1))\ell^2}{n}+A_{r,\ell}}.
%}
%This follows from the following construction. 
We consider the first $N/2$ edges to be colored red and the remaining $N/2$ edges to be colored blue. We use Theorem \ref{rhoNk} and the red edges to construct a star $K_{1,\ell}$ centered at vertex 1, at a cost of $\frac{(2+o(1))c_\ell}{n}\leq \frac{(4+o(1))\ell^2}{n}$. Here $\ell$ is to be determined. We then use the blue edges to find a low cost copy of $K_r$ in the red neighborhood of $1$ at a cost of $A_{r,\ell}$. We need to be a little careful here as the graph induced by the blue edges is disturbed as $G_{\ell,1/2}$ i.e. is not $K_\ell$. (Note that this would not be an issue in the POM mode, as we could let the red edges be those incident with vertex 1.) Our inductive assumption is that we can find w.h.p. a copy of $K_r$ in $G_{\ell,p}$ of cost $\ell^{-d_r+o(1)}$, provided $p$ is a constant independent of $\ell,n$. The base case will be triangles as claimed in Remark \ref{remy}. Thus,
\beq{Arnb}{
A_{r+1,n}\leq \min_\ell \set{\frac{(4+o(1))\ell^2}{n}+\frac{1}{\ell^{d_r+o(1)}}}.
}
We optimize \eqref{Arnb} by choosing $\ell=n^{1/(d_r+2+o(1)}$ which gives
$$A_{r+1,n}\leq \frac{1}{n^{d_r+o(1)/(d_r+2)}}.$$
This gives us Theorem \ref{Krthm} with 
\beq{dr}{
d_{r+1}=\frac{d_r}{d_r+2}.
}
Because $d_3=4/7$, we get $d_4=2/9$ and $d_5=1/10$. Putting $d_r=\frac{1}{11\times 2^{r-5}-1}$ we get
$$\frac{1}{d_{r+1}}=1+22\times 2^{r-5}-2=11\times 2^{r-4}-1.$$

This completes the proof of Theorem \ref{Krthm}.
\section{Spanning Trees and Arborescences}
\subsection{Spanning Trees}
In this section we will prove Theorem \ref{thMST}.
\subsubsection{Spanning Tree Upper Bound}\label{STUB}
We describe an algorithm that finds a tree of expected cost strictly less than $2\z(3)$. It can be improved, but our result is not tight and we try for simplicity here. The upper bound that we compute here is in the ROM model.

To begin we choose $0<\a<1$ and $\b$ where $\a\b>1$. 

{\bf Algorithm BUYTREE}:
\begin{enumerate}[{\bf Step 1}]
\item Let $E_1=\set{e_1,e_2,\ldots,e_{\a N}}$. We go through $E_1$ in order and accept an edge if its length $X_e\leq \b/n$ and it does not make a cycle with already accepted edges. Let $F_1$ be the forest induced by the accepted edges.
\item $E_1$ induces a graph $\G_1$ that is distributed as $G_{n,m}$ where $m\approx \g n/2,\g=\a\b$. So w.h.p. $\G_1$ and hence $F_1$ has a giant (tree) component 
\beq{C1}{
\text{$C_1$ of size asymptotically equal to $\brac{1-\frac{x}{\g}}n$ and $0<x<1$ and $xe^{-x}=\g e^{-\g}$.}
}
 This follows from Erd\H{o}s and R\'enyi \cite{ER}. See also Chapter 2 of \cite{BOOK}.

Outside of $C_1$ there will w.h.p. be $\n_k$ components of size $k$ where 
\beq{nuk}{
|\n_k-n\frac{k^{k-2}}{k!}\g^{k-1}e^{-\g k}|\leq n^{2/3}\text{ for }k=1,2,\ldots,O(\log n).
}
This also follows from \cite{ER}

Complete the building of the tree by using the case $k=1$ of Theorem \ref{rhoNk} and the edges of $E(K_n)\setminus E_1$ to select a single cheap edge from each small component to the giant $C_1$. We can do this in the ROM model, attacking each 1-purchasing problem separately.
\end{enumerate}
Let us carefully examine the expected cost of the tree found by the above algorithm. We break this into $P_1+P_2$ where $P_i$ is the expected cost of the edges added in Step i=1,2. We observe first that
\beq{P1}{
P_1=\frac{\b (|F_1|)}{2n}\approx \frac{\b}{2}\brac{1-\frac{x}{\g}+\frac{x^2}{2\g}}.
}
{\bf Explanation of \eqref{P1}:} We claim that in Step 1, Algorithm BUYTREE purchases the following number of edges asymptotically:
\beq{nofedges}{
n\brac{1-\frac{x}{\g}+\frac{x^2}{2\g}}.
}
The expected cost of each edge is $\frac{\b}{2n}$ and so verifying \eqref{nofedges} will verify \eqref{P1}. For this we rely on the following identities:
\begin{align}
\frac{1}{\g}\sum_{k=1}^\infty\frac{k^{k-1}}{k!}(\g e^{-\g})^k= \begin{cases}1&\g<1.\\\frac{x}{\g}&\g>1.\end{cases}\label{nov}\\
\frac{1}{\g}\sum_{k=1}^\infty\frac{k^{k-2}}{k!}(k-1)(\g e^{-\g})^k= \begin{cases}\frac{\g}{2}&\g<1.\\\frac{x^2}{2\g}&\g>1.\end{cases}\label{noe}
\end{align}
When $\g<1$, equation \eqref{nov} is derived from the fact that the expected number of vertices on tree components is asymptotically $n$. When $\g>1$ we replace $(\g e^{-\g})^k$ by $(xe^{-x})^k$ to get the value $x/\g$.

When $\g<1$, equation \eqref{noe} is derived from the fact that the expected number of edges on tree components is asymptotically $\g n/2$.  When $\g>1$ we replace $(\g e^{-\g})^k$ by $(xe^{-x})^k$ to get the value $x^2/2\g$.

We see from \eqref{nov} that the number of edges in the giant component $C_1$ is w.h.p. asymptotically equal to $n\brac{1-\frac{x}{\g}}$. We then see from \eqref{noe} that w.h.p. the number of edges selected that are in small trees is asymptotically equal to $nx^2/2\g$. There are w.h.p. $o(n)$ edges in $F_1$ that are obtained as spanning trees of small components with $s$ vertices and $\geq s$ edges.

As for the cost of Step 2, we have that 
\beq{P2}{
P_2\approx \frac{2}{(1-\a)\brac{1-\frac{x}{\g}}}\sum_{k=1}^\infty\frac{k^{k-3}}{k!}\g^{k-1}e^{-\g k}.
}
{\bf Explanation of \eqref{P2}:} The algorithm seeks a low cost edge between each small tree $T$ in $F_1$ to the giant component $C_1$. Observe first that there are w.h.p. 
$$m_1\approx (1-\a)|C_1|(n-|C_1|)$$ 
edges between $C_1$ and $\bar{C}_1=[n]\setminus C_1$. Indeed, for a given small constant $\e>0$ and set $S\subseteq [n]$ with $\e n\leq |S|\leq (1-\e)n$ let $\bar{e}(S)$ be the number of edges in $G_{n,(1-\a)}$ that belong to $S:\bar{S}$, the edges of $K_n$ between $S$ and $\bar{S}$. Then the Chernoff bounds \eqref{chern1}, \eqref{chern2} imply that
\beq{Sedge}{
\Pr(|\bar{e}(S)-(1-\a)|S|(n-|S|)|\geq n^{3/2}\log n)\leq 2\exp\set{-\frac{n^3(\log n)^2}{3|S|(n-|S|)(1-\a)}}\leq e^{-n(\log n)^2}.
}
The RHS of \eqref{Sedge} can be inflated by $2^n$ to account for all possible choices of $S$ and so w.h.p. $m_1$ is as claimed.

Note next that these $m_1$ edges are distributed uniformly over the set of $|C_1|(n-|C_1|)$ edges in $K_n$ with one endpoint in $|C_1|$. This is because given the fact that there are no $\G_1$ edges between $C_1$ and $\bar{C}_1$, we can interchange $E_1$ edges with non-$E_1$ edges within $S:\bar{S}$ without changing $F_1$. So, each such set of $m_1$ edges is equally likely. Under these circumstances, the number of edges between a small tree of size $k$ and $C_1$ is distributed as a hypergeometric with mean $(1-\a)k|C_1|\pm O(n^{1/2}\log n)$, given \eqref{Sedge}. As such, this will be concentrated around its mean (see Section 6 of Hoeffding \cite{Hoef}). Thus, using the case $k=1$ of Theorem \ref{rhoNk}, we see that the expected cost of connecting a tree with $k$ vertices to $C_1$ is at most $\approx\frac{2}{(1-\a)k|C_1|}$. Equation \eqref{P2} now follows from \eqref{C1} and \eqref{nuk}.

Putting $\a=0.69$ and $\b=3.5$ gives a total cost of less than $2.31$ which is less than $2\z(3)$. This verifies the upper bound in Theorem \ref{thMST}.

This bound can be improved by 
\begin{enumerate}
\item Using a sequence $\a_1,\b_2,\a_2,\b_2,\ldots,$ 
\item Using edges between small components in Step 2.
\end{enumerate}
Finally, in the AOM model, we can appeal to Theorem \ref{TSP} and find a Hamilton path at expected cost $O(1)$.
\subsubsection{Spanning Tree Lower Bound}\label{STLB}
In this section, we use the result of Theorem \ref{average} to obtain a lower bound on the expected cost of purchasing a spanning tree. Suppose we have an algorithm $A$. We apply it first to the edges contained in $S_1=[n+1]\setminus \set{x}$ where $x$ is chosen uniformly at random from $[n+1]$. Suppose that $A$ produces a spanning tree $T_1$ of $S_1$. Now use the algorithm of Section \ref{more} to find two edges $e_1,e_2$ from $x$ to $S_1$ to create a set of edges $T_2$. Now consider a fixed vertex $v$. Its degree in the random unicyclic graph $T_2$ is at least one and averages two. Applying Theorem \ref{average} we see that the expected cost of $T_2$ is at least 1.25, since the cost of each edge is counted exactly twice here. It follows that the expected cost of $T_1$ is at least $1.25-O(1/N)$ and the lower bound in Theorem \ref{thMST} follows.
\subsubsection{Improved Spanning Tree Lower Bound}\label{STLB1}
The proof we outline here is from Aldous, Angel and Berestycki \cite{AAB}. It sharpens the lower bound in Theorem \ref{average}. Consider an online algorithm.  As in the proof of Theorem \ref{average}, let $T$ be the index of the first selected item and let $q_k = \Pr(T > k)$ be the probability that the first $k$ edges are all rejected. This sequence is decreasing from $q_0=1$ to
$q_n=0$. The threshold $\theta_k$ for accepting the $k$'th edge is given by $\theta_k =\frac{q_{k-1}-q_k}{q_{k-1}}$. Conditional on $\{T \geq k\}$, the expected cost of accepting item $k$   equals $\theta_k^2/2$.
 
On the event $\{T = k\}$ let $a_k$ be the conditional expected number of subsequent edges accepted.  The conditional expected cost of the subsequent edges is at least $a_k\times \frac{a_k}{2(n-k)}$, since we must accept edges of cost $\leq\frac{a_k}{n-k}$ to get $a_k$ edges on average. Thus 
 \[ \E \mbox{(total  cost )} \ge \sum_{k=1}^n (q_{k-1}-q_k)\left(\frac{q_{k-1}-q_k}{2q_{k-1}} +    \frac{a_k^2}{2(n-k)} \right). \]
Any choice of 
\beq{con1}{
1 = q_0 \geq q_1 \geq \ldots \geq q_n = 0; \quad  0 \le a_k \le n-k, \ 1 \le k \le n
}
is feasible, and the constraint $\E |S| = 2$ becomes the constraint 
\begin{equation}\label{con2}
\sum_{k=1}^n (q_{k-1}-q_k)a_k  = 1.
\end{equation}
Thus if $\f_n$ is as in Theorem \ref{average} then
\[\f_n \geq\min \sum_{k=1}^n (q_{k-1}-q_k)\left(\frac{q_{k-1}-q_k}{2q_{k-1}} + \frac{a_k^2}{2(n-k)} \right) \]
minimized over $(q_k)$ and $(a_k)$ satisfying \eqref{con1}, \eqref{con2}. A somewhat difficult analysis gives $\f_n\geq 2.73747$.  

This completes the proof of Theorem \ref{thMST}.
\subsection{Spanning Arborescence}
In this section, we prove Theorem \ref{tharb}. 
\subsubsection{Spanning Arborescence Upper Bound in ROM}\label{arup}
We let $\e=1/\log n$ and consider the first $(1-\e)n(n-1)$ edges of $\dK_n$ to be colored red and the remaining $\e n(n-1)$ edges to be colored blue. Then for each vertex $v$ we use the red edges and case $k=1$ of Theorem \ref{rhoNk} to construct a random mapping digraph $D_f$ with vertex set $[n]$ and edges $\set{(v,f(v):v\in [n])}$ where $f(v)$ is $v$'s selection. The total expected cost of these edges is $\approx 2$. 

\begin{remark}\label{remD}
It will be important to examine the correlation between the edges of the digraph $D_f$ and the blue edges. What we claim is that one can see from the construction of $f(v)$ for $v\in [n]$ that if there are $k_v$ blue edges directed out of $v$, then these form a uniform random choice from $\binom{[n]\setminus\set{v}}{k_v}$. Furthermore, $B_v,B_w$ will be independent if $v\neq w$.
\end{remark}

It is known, see Chapter 15 of \cite{BOOK} or Chapter 14 of \cite{Boll} that if $Z$ is the number of components of $D_f$ then $\E(Z)\approx\Var(Z)\approx \frac12\log n$. Thus, the Chebyshev inequality implies that $Z$ is concentrated around its mean. We will however need an upper bound on $Z$ that holds with probability $1-o(1/n^3)$. It is known (see e.g. Theorem 15.1 of \cite{BOOK}) that the probability generating function (p.g.f.) $G_n$ of $Z$ is given by $G_n(x)=\E(x^Z)=\frac{x(x+1)\cdots(x+n-1)}{n!}$. So, for any positive integer $u$ we have that 
$$\Pr(Z\geq u)=\Pr(2^Z\geq 2^u)\leq \frac{\E(2^Z)}{2^u}= \frac{2\cdot3\cdots(2+n-1)}{2^un!}=\frac{(n+1)!}{2^un!}=\frac{n+1}{2^u}.$$
We take $u=(\log n)^2$ and find that $\Pr(Z\geq(\log n)^2)=o(n^{-3})$ as required.

The digraph $D_f$ consists of $Z$ digraphs $D_1,D_2,\ldots,D_Z$, each of which can be described as a directed cycle $C_i,i=1,2,\ldots,Z$ with oriented trees attached to each vertex of each cycle. The edges of the trees are oriented towards the cycle. We first delete an edge $(x_i,y_i)$ from each cycle. We are now left with $Z$ trees rooted at $X=\set{\r_1,\r_2,\ldots,\r_Z}$. We will use the blue edges to merge the trees into a spanning arborescence. We do this by adding a blue edge from the root of one of the trees to a vertex in another tree. This has the effect of reducing the number of trees by one.

We go through the blue edges in the given random order. Suppose we have examined $t-1$ blue edges. Let $X(t)$ denote the current set of roots of components. We stop when $|X(t)|=1$. For $\r\in X(t)$ we let $V_\r$ denote the set of vertices of the tree that contains $\r$. If the $t$th edge $e=(\r,\s)$ is such that (i) $\r\in X$, (ii) $\s\notin V_\s$ and (iii) the cost $x_e\leq n^{-3/4}$ then we purchase $e$ and reduce $X(t)$ by one. If we are successful in reducing $|X(t)|$ to one, in this way, then we will entail an additional cost of $(\log n)^2\times n^{-3/4}$, which is negligible.

Now suppose that $|X(t)|=k$ and that the components are $T_1,T_2,\ldots,T_k$ and the tree sizes are $m_1,m_2,\ldots,m_k$. Then the probability that the $t$th blue edge $e=(\r,\s)$ is {\em good} i.e. joins a root to a vertex in a different tree is at least
$$\frac{n-2-m_1}{n^{11/4}}+\cdots+\frac{n-2-m_k}{n^{11/4}}= \frac{k-1}{n^{7/4}}-\frac{2k}{n^{11/4}}\geq \frac{k-1}{2n^{7/4}}.$$
To see this, observe that there is a $1/n$ chance that $\r$ is the root of $T_1$. There is then the probability that $\t\notin T_1$ is at least $(n-2-m_1)/(n-2)$ (-2 as opposed to -1 from avoiding $f(\r)$.) Finally, there is an $n^{-3/4}$ chance that the cost is at most $n^{-3/4}$.

So, with probability $1-o(n^{-3})$, the number of blue edges needed is dominated by the sum $\Upsilon$ of $\ell=(\log n)^2$ independent geometric random variables with success probabilities $\l_k=\frac{k-1}{2n^{7/4}},k=2,3,\ldots,(\log n)^2$. The geometric random variable $Geo(\l)$ has p.g.f. $\frac{\l x}{1-(1-\l)x}$. The p.g.f. of $\Upsilon$ is therefore $\prod_{k=1}^\ell\frac{\l_kx}{1-(1-\l_k)x}$. So, for any $u>0$,
\begin{multline}
\Pr(\Upsilon\geq u)=\Pr\brac{\brac{1+\frac{1}{4n^{7/4}}}^\Upsilon\geq \brac{1+\frac{1}{4n^{7/4}}}^u}\leq\\ 
\brac{1+\frac{1}{4n^{7/4}}}^{-u}\prod_{k=2}^\ell \frac{\l_k\brac{1+\frac{1}{4n^{7/4}}}}{1-\brac{1-\l_k}\brac{1+\frac{1}{4n^{7/4}}}}\leq \exp\set{-\frac{u}{5n^{7/4}}}\prod_{k=2}^\ell\frac{1}{1-\frac{1}{\l_k(4n^{7/4}+1)}}\\
\leq \exp\set{-\frac{u}{5n^{7/4}}}\prod_{k=2}^\ell\brac{1+\frac{1}{k-1}}\leq \exp\set{-\frac{u}{5n^{7/4}}+1+\log\ell}.
\end{multline}
Putting $u=100\log n$ we see that with probability $1-o(n^{-3})$, we only need $O(n^{7/4}\log n)$ blue edges. On the other hand we have $\Omega(n^2/\log n)$ available, with this probability.

Finally, in the AOM model, we can appeal to the directed case of Theorem \ref{TSP} and find a Hamilton path at cost $O(1)$, w.h.p.
\subsubsection{Spanning Arborescence Lower Bound in RAM}\label{alb}
For the lower bound we color the edges red and blue as in Section \ref{arup}. Then if we want to solve the one-purchase problem of Section \ref{k1}, we assume the edges out of vertex 1 have the costs $x_1,x_2,\ldots,x_N$ where $N$ is the blue out-degree of vertex 1. An algorithm for finding an arborescence will provide a solution to the one-purchase problem, unless 1 is the root of the arborescence, which happens with probability $1/n$. In this case we just purchase a red edge of cost $\leq 2\log n/n$. If the expected cost of the arborescence found was $c<2$, then we would have a solution to the one-purchase problem with expected cost at most $c+2\log n/n^2$.

This completes the proof of Theorem \ref{tharb}.
\section{Perfect matchings}
We deal with bipartite and non-bipartite separately.
\subsection{Perfect Matchings in $K_{n,n}$}
In this section, we prove Theorem \ref{BPM}.
\subsubsection{Perfect Matching Lower Bound} 
For the lower bound observe that if $\m_n\leq c<2$ for some constant $c$ then the average cost of each edge in an optimal algorithm is at most $c/n$. But then, just as in Section \ref{alb}, we could use this to give an algorithm to improve the upper bound for $k=1$ in Theorem \ref{rhoNk}. Given $x_1,x_2,\ldots,x_n$ we would simply make these values the costs of the edges incident to vertex $1\in U$. The matching edge incident with vertex 1 would have expectation $c/n$.
\subsubsection{Perfect Matching Upper Bound}\label{BPM0}
For the upper bound we will use the following result of Walkup \cite{Walk}: Suppose that we label the partition of the vertex set of $K_{n,n}$ as $U,V$. Let $B_{k-out}$ denote the random bipartite graph where each vertex independently chooses $k$ random neighbors from the opposite part of the bipartition giving a graph with $2kn-O(1)$ distinct edges in expectation. Note that $B_{k-out}$ is chosen uniformly from some set of bipartite digraphs $\Omega_k$. Then
\beq{Wres}{
\Pr(B_{3-out}\text{ does not have a perfect matching})=o\bfrac{1}{n}.
}
Walkup actually proved that $B_{2-out}$ has a perfect matching w.h.p., but the failure probability is too high for our application.

To apply this result, we build something close to $B_{3-out}$ as follows: As we see an edge, we color it red or blue with probability 1/2. A vertex $v\in U$ uses the red edges incident with it and the algorithm of Theorem \ref{rhoNk} to choose three edges incident with it. A vertex $w\in V$ uses the blue edges incident with it and the algorithm of Theorem \ref{rhoNk} to choose three edges incident with it. Let $H$ be the bipartite graph created. 
\beq{costH}{
\text{The expected total cost of the edges in $H$ is $\approx 2n\times c_3/(n/2)=4c_3$}
}
as claimed in Theorem \ref{BPM}. It will have a perfect matching with probability $1-o(1/n)$. 

There is a small point to clarify. The bipartite graph produced by the algorithm is $B_{3-out}$ with the condition that no edge is chosen by both of its endpoints. The expected number of edges chosen twice in $B_{3-out}$ is at most $n^2\times (2/n)^2=6$. By computing higher moments we see that the probability no edge is chosen twice is $\approx e^{-6}$ and so the probability there is no perfect matching in the algorithm's graph is at most $e^6+o(1)$ times the probability there is no perfect matching in $B_{3-out}$.

This completes the proof of Theorem \ref{BPM}.
\subsection{Perfect Matchings in $K_{n}$}
We can assume that $n=2m$ is even. We can then look for a perfect matching between $[1,m]$ and $[m+1,2m]$. For this we can use the approach of Section \ref{BPM0} and this will give us an upper bound of $\approx 4c_3$. For a lower bound we can use $\approx n/2\times 2/(n/2)=2$, since we could use any algorithm for finding a perfect matching to find a solution to the one-purchase problem in Section \ref{k1} as we did for the lower bound in Section \ref{BPM}. 

This completes the proof of Theorem \ref{PM}.
\section{The Traveling Salesperson Problem}
\subsection{TSP Upper Bound} 
For this we first replace Walkup's result by the following: let $G_{k-out}$ be the random graph constructed by allowing each vertex to independently choose $k$ random neighbors. It is known, Bohman Frieze \cite{BF} that w.h.p. $G_{k-out}$ is Hamiltonian w.h.p. if $k\geq 3$. We want this probability to be $1-o(1/n)$ and to be sure of this we can use the result of Frieze \cite{FA} where we take $k=10$. 

We then write the uniform $[0,1]$ random variables $X_e$ as $X_e=\min\set{Z_{e,j}:j=1,2,\ldots,10}$ where the $Z_{e,j}$ are independently distributed as the random variable $Z\in[0,1]$ where $\Pr(Z\geq x)=(1-x)^{10}$. Assume first that when we examine an edge, we see these 10 values. Next, for $j=1,2,\ldots,10$ we use Corollary \ref{cor1} to choose an edge $e_{v,j}$ for each $v\in [n]$. The edges chosen will create a graph distributed as $G_{10-out}$ and the total cost will be at most $\approx10\times10\times2=200$. The first 10 arises because we do this for 10 values of $j$. The next 10 arises because the density of $Z$ near zero is $10x+o(x)$ and the 2 arises because $c_1=2$.

Of course when we examine an edge, we only see one value, $\xi$ say. To get around this, we generate another 9 values of $Z$, viz. $Z_2,Z_3,\ldots,Z_{10}$, but we condition here on $Z_j\geq \xi,j=2,3,\ldots,10$.  
\subsection{TSP Lower Bound} 
For a lower bound, we see that we could use an algorithm for finding a low cost Hamilton cycle to find a solution to the 2-purchase problem in Theorem \ref{rhoNk}. 

This completes the proof of Theorem \ref{TSP}.
\subsection{Directed case of TSP}
In this case we replace the results of \cite{BF}, \cite{FA} by the result of Cooper and Frieze \cite{CF}. We show in this paper that w.h.p. the random digraph $G_{2-in,2-out}$ is Hamiltonian. Here each vertex independently chooses 2 random out-neighbors and 2 random in-neighbors. To apply this, we replace each uniform edge-cost $X_e.e=(u,v)$ by $\min\set{Z_{e,j}:j=1,2,\ldots,4}$ where the $Z_{e,j}$ are independently distributed as the random variable $Z\in[0,1]$ where $\Pr(Z\geq x)=(1-x)^{4}$. We will then use $Z_{e,1},Z_{e,2}$ to give 2 random out-neighbors to $u$ and $Z_{e,3},Z_{e,4}$ to give 2 random in-neighbors to $v$. The rest of the argument is as for the undirected case. For the lower bound we see that each vertex chooses an in-neighbor and an out-neighbor. For the upper bound we replace $\approx10\times10\times2$ by $\approx 2\times2\times c_2$. We cannot claim a bound in terms of expectation because the proofs in \cite{CF} and the related \cite{CF1} do not give a small enough probability of failure.

This completes the proof of Theorem \ref{DTSP}.
\section{Final Remarks}
We have described several problems that can be analyzed within our framework. It would be of some interest to
\begin{enumerate}
\item Tighten the bounds, especially for minimum spanning trees. 
\item Replace $K_n$ by other graphs.
\item Exend the analysis to hypergraph structures.
\end{enumerate}
{\bf Acknowledgement:} We thank Colin Cooper for his comments and his contribution to the proof of Theorem \ref{rhoNk}.

\end{document}